\documentclass[a4paper, 10pt, twosides]{amsart}

\usepackage{lmodern}
\usepackage{enumitem}
\setenumerate{nolistsep}
\usepackage{dsfont}
\usepackage[usenames, dvipsnames]{color}
\usepackage{subfig}
\usepackage{todonotes}
\usepackage[colorlinks=true,
		linkcolor=MidnightBlue,
		citecolor=Mahogany,
		urlcolor=ForestGreen,
		pdfauthor={Atreyee Kundu and Debasish Chatterjee},
		pdftitle={ISS of continuous-time switched systems},
		pdfkeywords={switched systems, input-to-state stability},
		pdfsubject={Technical Report},
		plainpages=false]{hyperref}

\usepackage{mathtools}
\usepackage{mathabx}

\allowdisplaybreaks
\newtheorem{theorem}{Theorem}[section]

\newtheorem{lemma}[theorem]{Lemma}

\theoremstyle{definition}
\newtheorem{definition}[theorem]{Definition}

\theoremstyle{remark}
\newtheorem{remark}[theorem]{Remark}

\theoremstyle{assumption}
\newtheorem{assumption}[theorem]{Assumption}

\theoremstyle{fact}

\theoremstyle{claim}

\theoremstyle{proposition}
\newtheorem{proposition}[theorem]{Proposition}

\theoremstyle{cor}

\theoremstyle{algorithm}

\theoremstyle{convention}

\allowdisplaybreaks

\numberwithin{equation}{section}
\newcommand{\norm}[1]{\left\lVert{#1}\right\rVert} 
\newcommand{\abs}[1]{\left\lvert{#1}\right\rvert} 

\newcommand{\pmat}[1]{\begin{pmatrix}#1\end{pmatrix}}

\newcommand{\R}{\mathds{R}} 
\renewcommand{\P}{\mathcal{P}}  

\newcommand{\Ntsigma}{\mathrm{N}_{\sigma}}

\newcommand{\K}{\mathcal{K}}
\newcommand{\Kinfty}{\mathcal{K}_{\infty}}
\newcommand{\tendsto}{\rightarrow}
\newcommand{\KL}{\mathcal{KL}}

\newcommand{\lra}{\longrightarrow}

\renewcommand{\Ntsigma}{\mathrm{N}_{\sigma}(0,t)}
\newcommand{\Ntksigma}{\Ntsigma}
\newcommand{\supnormterm}{\gamma\bigl(\norm{v}_{[0,t]}\bigr)}
\newcommand{\ilsum}{\sum_{\substack{i=0\\\tau_{\Ntsigma+1}:=t}}^{\Ntsigma}}
\newcommand{\imsum}{\sum_{i=0}^{\Ntsigma-1}}
\newcommand{\muiterm}{\mu_{\sigma(\tau_{i})\sigma(\tau_{i+1})}}
\newcommand{\klsum}{\sum_{\substack{k=i+1\\\tau_{\Ntksigma+1}:=t}}^{\Ntsigma}}
\newcommand{\mukterm}{\mu_{\sigma(\tau_{k})\sigma(\tau_{k+1})}}
\newcommand{\kmsum}{\sum_{k=i+1}^{\Ntsigma-1}}
\newcommand{\ijsum}{\sum_{\substack{i:\sigma(\tau_{i})=j\\i=0,\cdots,\Ntsigma\\\tau_{\Ntsigma+1}:=t}}}
\newcommand{\iksum}{\sum_{\substack{i:\sigma(\tau_{i})=k\\i=0,\cdots,\Ntsigma\\\tau_{\Ntsigma+1}:=t}}}
\newcommand{\mnsum}{\displaystyle\sum_{(m,n)\in E(\P)}}
\newcommand{\mncount}{\#\{m\rightarrow n\}}

\newcommand{\stmeasure}{\abs{]s,t]\cap\biggl(\bigcup_{\displaystyle\substack{i=0\\\sigma(\tau_{i})=j}}^{+\infty}]\tau_{i},\tau_{i+1}]\biggr)}}
\newcommand{\ztmeasure}{\abs{]0,t]\cap\biggl(\bigcup_{\displaystyle\substack{i=0\\\sigma(\tau_{i})=j}}^{+\infty}]\tau_{i},\tau_{i+1}]\biggr)}}
\newcommand{\ztkmeasure}{\abs{]0,t]\cap\biggl(\bigcup_{\displaystyle\substack{i=0\\\sigma(\tau_{i})=k}}^{+\infty}]\tau_{i},\tau_{i+1}]\biggr)}}

\begin{document}
\title[ISS of continuous-time switched systems]{Generalized switching signals for input-to-state stability of switched systems}

\author[A.\ Kundu]{Atreyee Kundu}
\address{Systems \& Control Engineering, Indian Institute of Technology Bombay, Mumbai~--~400076, India}
\email[A.\ Kundu]{atreyee@sc.iitb.ac.in}
\author[D.\ Chatterjee]{Debasish Chatterjee}
\address{Systems \& Control Engineering, Indian Institute of Technology Bombay, Mumbai~--~400076, India}
\email[D.\ Chatterjee]{dchatter@iitb.ac.in}
\author[D.\ Liberzon]{Daniel Liberzon}
\address{Coordinated Science Laboratory, University of Illinois at Urbana-Champaign, Urbana IL 61801, USA}
\email[D.\ Liberzon]{liberzon@illinois.edu}
\date{\today}

\begin{abstract}
    This article deals with input-to-state stability (ISS) of continuous-time switched nonlinear systems. Given a family of systems with exogenous inputs such that not all systems in the family are ISS, we characterize a new and general class of switching signals under which the resulting switched system is ISS. Our stabilizing switching signals allow the number of switches to grow faster than an affine function of the length of a time interval, unlike in the case of average dwell time switching. We also recast a subclass of average dwell time switching signals in our setting and establish analogs of two representative prior results.
\end{abstract}

\maketitle

\section{Introduction}
\label{s:intro}
    A \emph{switched system} comprises of two components --- a family of systems and a switching signal. The \emph{switching signal} selects an active subsystem at every instant of time, i.e., the system from the family that is currently being followed \cite[\S 1.1.2]{Liberzon}. Stability of switched systems is broadly classified into two categories --- \emph{stability under arbitrary switching} \cite[Chapter 2]{Liberzon} and \emph{stability under constrained switching} \cite[Chapter 3]{Liberzon}. In the former category, conditions on the family of systems are identified such that the resulting switched system is stable under all admissible switching signals; in the latter category, given a family of systems, conditions on the switching signals are identified such that the resulting switched system is stable. In this article our focus is on stability of switched systems with exogenous inputs under constrained switching.

    Prior study in the direction of stability under constrained switching primarily utilizes the concept of \emph{slow switching} vis-a-vis \emph{(average) dwell time switching}. Exponential stability of a switched linear system under \emph{dwell time switching} was studied in \cite{Morse1996}. In \cite{XieWenLi2001} the authors showed that a switched nonlinear system is ISS under dwell time switching if all subsystems are ISS. A class of state-dependent switching signals obeying dwell time property under which a switched nonlinear system is integral input-to-state stable (iISS) was proposed in \cite{DePersis2003}. The dwell time requirement for stability was relaxed to \emph{average dwell time switching} to switched linear systems with inputs and switched nonlinear systems without inputs in \cite{HespanhaMorse}. ISS of switched nonlinear systems under average dwell time was studied in \cite{chatterjee07}. It was shown that if the individual subsystems are ISS and their ISS-Lyapunov functions satisfy suitable conditions, then the switched system has the ISS, exponentially-weighted ISS, and exponentially-weighted iISS properties under switching signals obeying sufficiently large average dwell time. Given a family of systems such that not all systems in the family are ISS, it was shown in the recent work \cite{Yang14} that it is possible to construct a class of hybrid Lyapunov functions to guarantee ISS of the switched system provided that the switching signal neither switches too frequently nor activates the non-ISS subsystems for too long. In \cite{Liberzon_IOSS} input/output-to-state stability (IOSS) of switched nonlinear systems with families in which not all subsystems are IOSS, was studied. It was shown that the switched system is IOSS under a class of switching signals obeying \emph{average dwell time} property and constrained point-wise activation of unstable systems.
%
%

    Given a family of systems, possibly containing non-ISS dynamics, in this article we study ISS of switched systems under switching signals that transcend beyond the average dwell time regime in the sense that the number of switches on every interval of time can grow faster than an affine function of the length of the interval. Our characterization of stabilizing switching signals involve pointwise constraints on the duration of activation of the ISS and non-ISS systems, and the number of occurrences of the admissible switches, certain pointwise properties of the quantities defining the above constraints, and a summability condition. In particular, our contributions are:
    \begin{itemize}[label=$\circ$,leftmargin=*]
        \item We allow non-ISS systems in the family and identify a class of switching signals under which the resulting switched system is ISS.
        \item Our class of stabilizing switching signals encompasses the average dwell time regime in the sense that on every interval of time the number of switches is allowed to grow faster than an affine function of the length of the interval. Earlier in \cite{TACsub} we proposed a class of switching signals beyond the average dwell time regime for global asymptotic stability of continuous-time switched nonlinear systems.
        \item Although this is not the first instance when non-ISS subsystems are considered (see e.g., \cite{Liberzon_IOSS,Yang14}), to the best of our knowledge, this is the first instance when non-ISS subsystems are considered and the proposed class of stabilizing switching signals goes beyond the average dwell time condition.
        \item We recast a subclass of average dwell time switching signals in our setting and establish analogs of an ISS version of \cite[Theorem 2]{Liberzon_IOSS}, and \cite[Theorem 3.1]{chatterjee07} as two corollaries of our main result.
    \end{itemize}

    The remainder of this article is organized as follows: In \S\ref{s:prelims} we formulate the problem under consideration and catalog certain properties of the family of systems and the switching signal. Our main results appear in \S\ref{s:mainres}, and we provide a numerical example illustrating our main result in \S\ref{s:ex}. In \S\ref{s:discussn} we recast prior results in our setting. The proofs of our main results are presented in a consolidated manner in \S\ref{s:proofs}.

    {\bf Notations}: Let $\R$ denote the set of real numbers, $\norm{\cdot}$ denote the Euclidean norm, and for any interval $I\subset[0,+\infty[$ we denote by $\norm{\cdot}_{I}$ the essential supremum norm of a map from $I$ into some Euclidean space. For measurable sets $A\subset\R$ we let $\abs{A}$ denote the Lebesgue measure of $A$.

\section{Preliminaries}
\label{s:prelims}
    We consider the \emph{switched system}
    \begin{align}
    \label{issc_e:swsys}
        \dot{x}(t) = f_{\sigma(t)}\bigl(x(t),v(t)\bigr),\:\:x(0) = x_{0}\:\text{(given)},\:\:t\geq 0.
    \end{align}
    generated by
    \begin{itemize}[label = $\circ$, leftmargin = *]
        \item a family of continuous-time systems with exogenous inputs
            \begin{align}
            \label{issc_e:family}
                \dot{x}(t) = f_{i}\bigl(x(t),v(t)\bigr),\:\:x(0) = x_{0}\:\text{(given)},\:\:i\in\P,\:\:t\geq 0,
            \end{align}
        where $x(t)\in\R^{d}$ is the vector of states and $v(t)\in\R^{m}$ is the vector of inputs at time $t$, $\P = \{1,2,\cdots,N\}$ is a finite index set, and
        \item a piecewise constant function $\sigma:[0,+\infty[\lra\P$ that selects, at each time $t$, the index of the active system from the family \eqref{issc_e:family}; this function $\sigma$ is called a \emph{switching signal}. By convention, $\sigma$ is assumed to be continuous from right and having limits from the left everywhere, and we call such switching signals admissible. We let $\mathcal{S}$ denote the set of all such admissible switching signals.
    \end{itemize}
    We assume that for each $i\in\P$, $f_{i}$ is locally Lipschitz, and $f_{i}(0,0) = 0$. Let the exogenous inputs $t\mapsto v(t)$ be Lebesgue measurable and essentially bounded; therefore, a solution to the switched system \eqref{issc_e:swsys} exists in the Carath$\acute{e}$odory sense \cite[Chapter 2]{Filippov} for some non-trivial time interval containing $0$.

    Given a family of systems \eqref{issc_e:family}, our focus is on identifying a class of switching signals $\sigma\in\mathcal{S}$ under which the switched system \eqref{issc_e:swsys} is ISS. Recall that
    \begin{definition}[{\cite[\S2]{chatterjee07}}]
    \label{issc_d:iss}
        The switched system \eqref{issc_e:swsys} is {input-to-state stable} (ISS) for a given $\sigma$ if there exist class $\Kinfty$ functions $\alpha,\chi$ and a class $\KL$ function $\beta$ such that for all inputs $v$ and initial states $x_{0}$, we have\footnote{$\mathcal{K} := \{\phi:[0,+\infty[\to[0,+\infty[\:\:\big|\:\:\phi\:\text{is continuous, strictly increasing,}\:\phi(0) = 0\}$, $\KL := \bigl\{\phi:[0,+\infty[^{2}\lra[0,+\infty[\:\:\big|\:\:\phi(\cdot,s)\in\K\:\:\text{for each}\:\:s\:\:\text{and}\:\:\phi(r,\cdot)\searrow 0\:\:\text{as}\:\:s\nearrow +\infty\:\:\text{for each}\:\:\:r\bigr\}$, $\displaystyle\Kinfty := \bigl\{\phi\in\K\:\:\big|\:\:\phi(r) \to +\infty\:\:\text{as}\:\:r\to+\infty\bigr\}$}
        \begin{align}
        \label{issc_e:iss}
            \alpha(\norm{x(t)}) \leq \beta(\norm{x_{0}},t) + \chi(\norm{v}_{[0,t]})\:\:\text{for all}\:\: t\geq 0.
        \end{align}
        If one can find $\alpha$, $\beta$ and $\chi$ such that \eqref{issc_e:iss} holds over a class $\mathcal{S}'$ of $\sigma$, then we say that \eqref{issc_e:swsys} is {uniformly ISS} over $\mathcal{S}'$.
    \end{definition}

     Note that if no input is present, i.e., $v\equiv 0$, then \eqref{issc_e:iss} reduces to GAS of \eqref{issc_e:swsys}. We next catalog certain properties of the family of systems \eqref{issc_e:family} and the switching signal $\sigma$. These properties will be required for our analysis towards deriving the class of stabilizing switching signals.
     \subsection{Properties of the family of systems}
     \label{issc_ss:familyprop}
     Let $\P_{S}$ and $\P_{U}\subset\P$ denote the sets of indices of ISS and non-ISS systems in the family \eqref{issc_e:family}, respectively, $\P = \P_{S}\sqcup\P_{U}$. Let $E(\P)$ be the set of all ordered pairs $(i,j)$ such that it is admissible to switch from system $i$ to system $j$, $i,j\in\P$.
        \begin{assumption}
        \label{issc_assumption:key}
            There exist class $\Kinfty$ functions $\underline{\alpha}$, $\overline{\alpha}$, $\gamma$, continuously differentiable functions $V_{i}:\R^{d}\lra[0,+\infty[$, $i\in\P$, and constants $\lambda_{i}\in\R$ with $\lambda_{i} > 0$ for $i\in\P_{S}$ and $\lambda_{i} < 0$ for $i\in\P_{U}$, such that for all $\xi\in\R^{d}$ and $\eta\in\R^{m}$, we have
            \begin{gather}
                \label{issc_e:Lyapineq}\underline{\alpha}(\norm{\xi}) \leq V_{i}(\xi) \leq \overline{\alpha}(\norm{\xi}),\\
                \label{issc_e:dLyapineq}\Biggl<\frac{\partial V_{i}}{\partial\xi}(\xi),f_{i}(\xi,\eta)\Biggr> \leq -\lambda_{i}V_{i}(\xi) + \gamma(\norm{\eta}).
            \end{gather}
        \end{assumption}
        \addtocounter{equation}{1}
        \begin{remark}
         Conditions \eqref{issc_e:Lyapineq} and \eqref{issc_e:dLyapineq} are equivalent to an ISS version of \cite[(7) and (18)]{Liberzon_IOSS}. The functions $V_{i}$'s are called the \emph{ISS-Lyapunov-like functions}, see \cite{Sontag95,Angeli99,Krichman01} for detailed discussion regarding the existence of such functions and their properties. In particular, condition \eqref{issc_e:dLyapineq} is equivalent to the ISS property for ISS subsystems \cite{Sontag95} and the unboundedness observability property for the non-ISS subsystems \cite{Krichman01}.
        \end{remark}

        \begin{assumption}
        \label{issc_assumption:muijreln}
            For each pair $(i,j)\in E(\P)$ there exist $\mu_{ij} > 0$ such that the ISS-Lyapunov-like functions are related as follows:
            \begin{align}
            \label{issc_e:muijineq}
                V_{j}(\xi) \leq \mu_{ij}V_{i}(\xi)\:\:\text{for all}\:\:\xi\in\R^{d}.
            \end{align}
        \end{assumption}
        \begin{remark}
            The assumption of linearly comparable Lyapunov-like functions, i.e., there exists $\mu\geq 1$ such that
            \begin{align}
            \label{e:muineq}
                V_{j}(\xi) \leq \mu V_{i}(\xi)\:\:\text{for all}\:\:\xi\in\R^{d}\:\:\text{and}\:\:i,j\in\P
            \end{align}
            is standard in the theory of stability under average dwell time switching \cite[Theorem 3.2]{Liberzon}; \eqref{issc_e:muijineq} affords sharper estimates compared to \eqref{e:muineq}.
        \end{remark}

    \subsection{Properties of the switching signal}
        \label{ss:swsigprop}
            Fix $t > 0$. For a switching signal $\sigma$ we let $N_{\sigma}(0,t)$ denote the number of switches on the interval $]0,t]$, and $0 =: \tau_{0}<\tau_{1}<\cdots<\tau_{\Ntsigma}$ denote the corresponding switching instants before (and including) $t$.
        \begin{itemize}[label = $\circ$, leftmargin = *]
            \item We let
            \begin{align}
            \label{e:holdtime}
                S_{i+1} := \tau_{i+1} - \tau_{i},\quad i = 0,1,\cdots,
            \end{align}
            denote the \emph{i-th holding time} of $\sigma$.
        \item On an interval $]s,t]\subset[0,+\infty[$ of time, let
            \begin{align}
                \label{e:activetimeiss} \mathrm{T}^{\mathrm{S}}_{j}(s,t) := \displaystyle \stmeasure,\\
                \intertext{and}
                \label{e:activetimenoniss} \mathrm{T}^{\mathrm{U}}_{k}(s,t) := \abs{]s,t]\cap\biggl(\bigcup_{\displaystyle\substack{i=0\\\sigma(\tau_{i})=k}}^{+\infty}]\tau_{i},\tau_{i+1}]\biggr)}
            \end{align}
            denote the \emph{duration of activation} of a system $j\in\P_{S}$ and $k\in\P_{U}$, respectively. Clearly, $\displaystyle{\sum_{k\in\P_{U}}\mathrm{T}^{\mathrm{U}}_{k}(s,t) + \sum_{j\in\P_{S}}\mathrm{T}^{\mathrm{S}}_{j}(s,t)} = t-s\:\:\text{for all}\:\: 0\leq s<t<+\infty$.
        \item For a pair $(m,n)\in E(\P)$ let
            \begin{align}
            \addtocounter{equation}{1}
            \label{e:swnumber}
                \mathrm{N}_{mn}(s,t) := \#\{m\rightarrow n\}_{s}^{t}
            \end{align}
            be the \emph{number of switches} from system $m$ to system $n$ on the interval $]s,t]\subset[0,+\infty[$ of time. We have the immediate identity: $\Ntsigma = \sum_{(m,n)\in E(\P)}\mathrm{N}_{mn}(0,t)$, $t > 0$.
        \end{itemize}
        In the sequel we require the following class of functions:
        \begin{definition}
        \label{issct_d:funcdefinition}
            A function $\varrho:[0,+\infty[^{2}\rightarrow [0,+\infty[$ belongs to class $\mathcal{F}\mathcal{K}_{\infty}$ if
            \begin{itemize}[label = $\circ$, leftmargin = *]
                \item $\varrho$ is continuous, and
                \item for every fixed first argument, $\varrho$ is in class $\Kinfty$ in the second argument.
            \end{itemize}
        \end{definition}
            \begin{assumption}
            \label{issc_assumption:swsigbounds}
                There exist class $\mathcal{F}\mathcal{K}_{\infty}$ functions $\rho^{\mathrm{S}}_j$, ${j\in \mathcal P_S}$, $\rho^{\mathrm{U}}_k$, ${k\in \mathcal P_U}$, $\rho_{mn}$, ${(m, n)\in E(\mathcal P)}$, and positive constants $\overline{\mathrm{T}}^{\mathrm{S}}_{j}$,${j\in\P_{S}}$, $\overline{\mathrm{T}}^{\mathrm{U}}_{k}$, ${k\in\P_{U}}$, $\overline{\mathrm{N}}_{mn}$, ${(m,n)\in E(\P)}$, such that on every interval $]s,t]\subset[0,+\infty[$ of time, the functions $\mathrm{T}^{\mathrm{S}}_{j}(s,t)$, ${j\in\P_{S}}$, $\mathrm{T}^{\mathrm{U}}_{k}(s,t)$, ${k\in\P_{U}}$, and $\mathrm{N}_{mn}(s,t)$, ${(m,n)\in E(\P)}$, defined in \eqref{e:activetimeiss}, \eqref{e:activetimenoniss}, and \eqref{e:swnumber}, respectively, satisfy the following inequalities:
        \begin{align}
            \label{issc_e:isslb} \mathrm{T}^{\mathrm{S}}_{j}(s,t) &\geq -\overline{\mathrm{T}}^{\mathrm{S}}_{j} + \rho^{\mathrm{S}}_{j}(s,t-s),\\
            \label{issc_e:nonissub} \mathrm{T}^{\mathrm{U}}_{k}(s,t) &\leq \overline{\mathrm{T}}^{\mathrm{U}}_{k} + \rho^{\mathrm{U}}_{k}(s,t-s),\\
            \label{issc_e:swub} \mathrm{N}_{mn}(s,t) &\leq \overline{\mathrm{N}}_{mn} + \rho_{mn}(s,t-s).
        \end{align}
    \end{assumption}
    \begin{remark}
    \label{a:kinffunctions}
        On every interval $]s,t]\subset[0,+\infty[$ of time, conditions \eqref{issc_e:isslb} and \eqref{issc_e:nonissub} constrain the duration of activation of a system $j\in\P_{S}$ and $k\in\P_{U}$, respectively, and condition \eqref{issc_e:swub} constrains the number of occurrences of an admissible switch $(m,n)\in E(\P)$. Each bound is provided in terms of a class $\mathcal{F}\mathcal{K}_{\infty}$ function (from $\rho^{\mathrm{S}}_{j}$, ${j\in\P_{S}}$, $\rho^{\mathrm{U}}_{k}$, ${k\in\P_{U}}$, $\rho_{mn}$, ${(m,n)\in E(\P)}$) and a positive offset (from $\overline{\mathrm{T}}^{\mathrm{S}}_{j}$,${j\in\P_{S}}$, $\overline{\mathrm{T}}^{\mathrm{U}}_{k}$, ${k\in\P_{U}}$, $\overline{\mathrm{N}}_{mn}$, ${(m,n)\in E(\P)}$). We consider point-wise lower bounds on the duration of activation of ISS subsystems, and upper bounds on the duration of activation of non-ISS subsystems and the number of occurrences of admissible switches on every interval $]s,t]\subset[0,+\infty[$ of time. In the analysis towards identifying switching signals for ISS of switched systems, such bounds are standard assumptions. For example, in \cite{Liberzon_IOSS} and \cite{chatterjee07} the number of switches on every interval of time is allowed to grow at most as an affine function of the length of the interval. In the presence of non-ISS subsystems, the duration of activation of such systems is also constrained on every interval of time in \cite{Liberzon_IOSS}. We use the class $\mathcal{F}\mathcal{K}_{\infty}$ functions $\rho^{\mathrm{S}}_{j}(s,t-s)$, ${j\in\P_{S}}$, $\rho^{\mathrm{U}}_{k}(s,t-s)$, ${k\in\P_{U}}$, and $\rho_{mn}(s,t-s)$, ${(m,n)\in E(\P)}$ with two arguments --- the initial value of the interval $s\in[0,+\infty[$ and the length of the interval $(t-s)$, with the objective to allow the number of switches on any interval of time to grow faster than the case of average dwell time switching as we shall see momentarily.
    \end{remark}


\section{Main Results}
\label{s:mainres}
     We are now in a position to present our main results.
    \begin{theorem}
    \label{issc_t:mainres1}
        Consider the family of systems \eqref{issc_e:family}. Let $\P_{S}, \P_{U}\subset \P$ and $E(\P)$ be as described in \S\ref{issc_ss:familyprop}. Suppose that Assumptions \ref{issc_assumption:key} and \ref{issc_assumption:muijreln} hold. Let there exist constants $c_{1}$ and $c_{2}$, and a class $\mathcal{F}\mathcal{K}_{\infty}$ function $\rho:[0,+\infty[^{2}\to[0,+\infty[$ satisfying $\rho(0,0) = 0$ such that the following conditions hold:
        \begin{align}
            \label{issc_e:maincondn1}&-\sum_{j\in\P_{S}}\abs{\lambda_{j}}\rho^{\mathrm{S}}_{j}(r,s)+\sum_{k\in\P_{U}}\abs{\lambda_{k}}\rho^{\mathrm{U}}_{k}(r,s)+\sum_{(m,n)\in E(\P)}{(\ln\mu_{mn})}\rho_{mn}(r,s) \leq c_{1}-\rho(r,s)\\
            &\text{for every interval $]r,r+s]\subset[0,+\infty[$ of time, and}\nonumber\\
            \label{issc_e:maincondn2}&\lim_{t\rightarrow+\infty}\sum_{i=0}^{\Ntsigma}\exp\bigl(-\rho(\tau_{i},t-\tau_{i})\bigr) \leq c_{2}.
        \end{align}
        Here $\lambda_{j}$, ${j\in\P_{S}}$, $\lambda_{k}$, ${k\in\P_{U}}$, and $\mu_{mn}$, ${(m,n)\in E(\P)}$ are as in \eqref{issc_e:dLyapineq} and \eqref{issc_e:muijineq}, respectively, and class $\mathcal{F}\mathcal{K}_{\infty}$ functions $\rho^{\mathrm{S}}_{j}$, ${j\in\P_{S}}$, $\rho^{\mathrm{U}}_{k}$, ${k\in\P_{U}}$, $\rho_{mn}$, ${(m,n)\in E(\P)}$ are as in Assumption \ref{issc_assumption:swsigbounds}. Then the switched system \eqref{issc_e:swsys} is uniformly input-to-state stable (ISS) for every $\sigma\in\mathcal{S}$ satisfying \eqref{issc_e:isslb}, \eqref{issc_e:nonissub}, and \eqref{issc_e:swub} for every interval $]r,r+s]\subset[0,+\infty[$ of time.
    \end{theorem}
        See \S\ref{s:proofs} for a detailed proof of the above theorem.
    \begin{remark}
    \label{r:reg3.1-a}
        The condition \eqref{issc_e:maincondn1} provides a point-wise upper bound on the difference between the weighted class $\mathcal{F}\mathcal{K}_{\infty}$ functions
       \[
        (r,s)\mapsto\displaystyle\sum_{k\in\P_{U}}\abs{\lambda_{k}}\rho^{\mathrm{U}}_{k}(r,s) + \sum_{(m,n)\in E(\P)}(\ln\mu_{mn})\rho_{mn}(r,s),\:\:\text{and}
        \]
        \[
            (r,s)\mapsto\displaystyle\sum_{j\in\P_{S}}\abs{\lambda_{j}}\rho^{\mathrm{S}}_{j}(r,s)
        \]
       in terms of a constant $c_{1}$ and another class $\mathcal{F}\mathcal{K}_{\infty}$ function $\rho$ satisfying $\rho(0,0) = 0$, where the class $\mathcal{F}\mathcal{K}_{\infty}$ functions $\rho^{\mathrm{S}}_{j}$, ${j\in\P_{S}}$, $\rho^{\mathrm{U}}_{k}$, ${k\in\P_{U}}$, and $\rho_{mn}$, ${(m,n)\in E(\P)}$ constrain the duration of activation of ISS subsystems and non-ISS subsystems, and the number of occurrences of the admissible switches, respectively on every interval $]r,r+s]\subset[0,+\infty[$ of time.
    \end{remark}

    \begin{remark}
         The condition \eqref{issc_e:maincondn2} deals with summability of a series with non-negative terms involving the class $\mathcal{F}\mathcal{K}_{\infty}$ function $\rho$ satisfying $\rho(0,0) = 0$, the number of switches $\Ntsigma$ before (and including) $t > 0$, and the corresponding switching instants $0 =: \tau_{0} < \tau_{1} < \cdots < \tau_{\Ntsigma}$.
    \end{remark}
    \begin{remark}
        The constants $c_{1}$ and $c_{2}$ on the right-hand sides of \eqref{issc_e:maincondn1} and \eqref{issc_e:maincondn2} ensure \emph{uniform} ISS of the switched system \eqref{issc_e:swsys} over all switching signals $\sigma$ satisfying \eqref{issc_e:isslb}, \eqref{issc_e:nonissub}, \eqref{issc_e:swub}, \eqref{issc_e:maincondn1} and \eqref{issc_e:maincondn2}.
     \end{remark}

     \begin{remark}
     \label{issct_r:adtcompa}
        Our class of stabilizing switching signals goes beyond the average dwell time regime in the sense that on every interval of time the number of switches is allowed to grow faster than an affine function of the length of the interval. We elaborate on this feature with the aid of the following example:\\
        Fix $t > 0$. Let us study how close to $t$ can the $\tau_{i}$'s be placed under condition \eqref{issc_e:swub}. We have $\Ntsigma \leq \mathrm{N}_{0} + \lfloor\rho_{\mathrm{N}(0,t)}\rfloor$,
        where $\displaystyle{\mathrm{N}_{0} := \sum_{(m,n)\in E(\P)}\overline{\mathrm{N}}_{mn}}$ and $\displaystyle{\rho_{\mathrm{N}}(0,t) := \sum_{(m,n)\in E(\P)}\rho_{mn}(0,t)}$.
        Consequently, $\displaystyle{\sum_{i=0}^{\Ntsigma}\exp\bigl(-\rho(\tau_{i},t-\tau_{i})\bigr)}$ is at most equal to $\displaystyle{\sum_{i=0}^{\mathrm{N}_{0}+\lfloor\rho_{\mathrm{N}}(0,t)\rfloor}\exp\bigl(-\rho(\tau_{i},t-\tau_{i})\bigr)}$.
        However small a time interval may be, at most $\mathrm{N}_{0}$ switches are allowed. So these many switches can be placed arbitrarily close to $t$. For the rest of the $\lfloor\rho_{\mathrm{N}}(0,t)\rfloor = n$ (say) switches that can be placed on $]0,t]$, we have
        \begin{itemize}[label = $\circ$, leftmargin = *]
            \item on the interval $]\tau_{n},t]$ at most $\mathrm{N}_{0} + 1$ switches are allowed,
            \item on the interval $]\tau_{n-1},t]$ at most $\mathrm{N}_{0} + 2$ switches are allowed,
            \item $\cdots$
        \end{itemize}
        Joint validity of the above conditions leads to
        \begin{align*}
            \tau_{n} &= t - \inf\{r > 0\:|\:\rho_{\mathrm{N}}(r,s) > 1\:\:\text{with}\:\:s = t-r\},\\
            \tau_{n-1} &= t - \inf\{r > 0\:|\:\rho_{\mathrm{N}}(r,s) > 2\:\:\text{with}\:\:s = t-r\},\\
            \cdots\\
            \tau_{2} &= t - \inf\{r > 0\:|\:\rho_{\mathrm{N}}(r,s) > (n-1)\:\:\text{with}\:\:s = t-r\},\\
            \tau_{1} &= t - \inf\{r > 0\:|\:\rho_{\mathrm{N}}(r,s) > n\:\:\text{with}\:\:s = t-r\},
        \end{align*}
        i.e.,
        \begin{align*}
            \tau_{n} &= t-\rho_{\mathrm{N}}^{-1}(\cdot,t-\cdot)(1),\\
            \tau_{n-1} &= t-\rho_{\mathrm{N}}^{-1}(\cdot,t-\cdot)(2),\\
            \cdots\\
            \tau_{2} &= t-\rho_{\mathrm{N}}^{-1}(\cdot,t-\cdot)(n-1),\\
            \tau_{1} &= t-\rho_{\mathrm{N}}^{-1}(\cdot,t-\cdot)(n).
        \end{align*}
        Now let us study the above phenomenon under average dwell time switching. Recall that \cite[p.\ 58]{Liberzon} a switching signal $\sigma$ has average dwell time $\tau_{a}$ if there exist two positive numbers $\mathrm{N}_{0}$ and $\tau_{a}$ such that $\mathrm{N}_{\sigma}(s,t) \leq \mathrm{N}_{0} + \frac{t-s}{\tau_{a}}$ for all $0\leq s < t$. Let the $\mathrm{N}_{0}$ switches be placed arbitrarily close to $t$ as already explained. As regard to the remaining $\lfloor\frac{t}{\tau_{a}}\rfloor$ switches,
        \begin{itemize}[label = $\circ$, leftmargin = *]
            \item on every interval of length $t - (t-n\tau_{a})$, at most $\mathrm{N}_{0}+n$ switches are allowed,
            \item on every interval of length $(t-(n-1)\tau_{a}) - (t-\tau_{a})$, at most $\mathrm{N}_{0}+1$ switches are allowed,
            \item $\cdots$
        \end{itemize}
        Consequently, we have
        \begin{align*}
            \tau_{n} &= t - \tau_{a},\\
            \tau_{n-1} &= t - 2\tau_{a},\\
            \cdots\\
            \tau_{2} &= t - (n-1)\tau_{a},\\
            \tau_{1} &= t - n\tau_{a}.
        \end{align*}
        As is evident from the above discussion, our class of switching signals allows number of switches on every interval of time to grow faster than an affine function of the length of the interval.
    \end{remark}

    We next consider two simple cases where both the functions $\rho$ and $\rho_{N}$ are such that for all $r_{1},r_{2} \geq 0$ and all $s > 0$
    \begin{align*}
        \rho(r_{1},s) = \rho(r_{2},s),\:\:\text{and}\:\:\rho_{\mathrm{N}}(r_{1},s) = \rho_{\mathrm{N}}(r_{2},s),
    \end{align*}
    and discuss boundedness of the quantity $\displaystyle\sum_{i=0}^{\Ntsigma}\exp\bigl(-\rho(\tau_{i},t-\tau_{i})\bigr)$ with $t$.
    \begin{lemma}
    \label{issc_lem:rhoaffine}
        Let
        \begin{align}
        \label{issc_e:rhoaffine}
            \rho(r,s) = k_{1}s + k_{2}\quad\text{for some}\:\:k_{1}, k_{2} > 0,\:\:s \geq 0,
        \end{align}
        and let $\rho_{N}$ be such that the switches be equispaced in time, i.e., they satisfy
        \begin{align*}
            \tau_{n} &= t - \rho_{\mathrm{N}}^{-1}(\cdot,t-\cdot)(1)\\
            \tau_{n-1} &= t - 2\rho_{\mathrm{N}}^{-1}(\cdot,t-\cdot)(1)\\
            \cdots\\
            \tau_{2} &= t - (n-1)\rho_{\mathrm{N}}^{-1}(\cdot,t-\cdot)(1)\\
            \tau_{1} &= t - n\rho_{\mathrm{N}}^{-1}(\cdot,t-\cdot)(1).
         \end{align*}
        Then $\displaystyle{\lim_{t\rightarrow+\infty}\sum_{i=0}^{\Ntsigma}\exp\bigl(-\rho(\tau_{i},t-\tau_{i})\bigr) < +\infty}$.
    \end{lemma}

    \begin{lemma}
    \label{issc_lem:rhothreehalves}
        Let
        \begin{align}
        \label{issc_e:rhothreehalves}
            \rho(r,s) = k_{1}s^{3/2} + k_{2}\quad\text{for some}\:\:k_{1}, k_{2} > 0,\:\:s \geq 0,
        \end{align}
        and let $\rho_{N}$ be such that the switches be equispaced in time, i.e., they satisfy
        \begin{align*}
            \tau_{n} &= t - \rho_{\mathrm{N}}^{-1}(\cdot,t-\cdot)(1)\\
            \tau_{n-1} &= t - 2\rho_{\mathrm{N}}^{-1}(\cdot,t-\cdot)(1)\\
            \cdots\\
            \tau_{2} &= t - (n-1)\rho_{\mathrm{N}}^{-1}(\cdot,t-\cdot)(1)\\
            \tau_{1} &= t - n\rho_{\mathrm{N}}^{-1}(\cdot,t-\cdot)(1).
         \end{align*}
        Then $\displaystyle{\lim_{t\rightarrow+\infty}\sum_{i=0}^{\Ntsigma}\exp\bigl(-\rho(\tau_{i},t-\tau_{i})\bigr) < +\infty}$.
    \end{lemma}
    The proofs of Lemmas \ref{issc_lem:rhoaffine} and \ref{issc_lem:rhothreehalves} are presented in \S\ref{s:proofs}.



\section{Numerical Example}
\label{s:ex}
    We consider $\P = \{1,2\}$ with
    \begin{align*}
        f_{1}(x,v) &= \pmat{-x_{1}+\sin(x_{1}-x_{2})\\-x_{2}+0.8\sin(x_{2}-x_{1})+0.5v},\\
        \intertext{and}
        f_{2}(x,v) &= \pmat{x_{1}+\sin(x_{1}-x_{2})\\x_{2}+\sin(x_{2}-x_{1})+0.5v}.
    \end{align*}
    Consequently, $\P_{S} = \{1\}$ and $\P_{U} = \{2\}$. Let $v\equiv 1$. With the choice $V_{1}(x) = 0.5(x_{1}^2 + 1.25x_{2}^2)$, $V_{2}(x) = 0.5(x_{1}^{2}+x_{2}^{2})$, we have $\lambda_{1} = 1.75$, $\lambda_{2} = -2.1667$, $\mu_{12} = 1$, and $\mu_{21} = 2$, see \cite[\S5]{Liberzon_IOSS} for a detailed discussion.

    Let $\rho(r,s) = k_{1}s^{3/2}+k_{2}$ with $k_{1},k_{2} > 0$. We have already shown in Lemma \ref{issc_lem:rhothreehalves} that with the above structure on $\rho$, the term $\displaystyle{\lim_{t\to +\infty}\sum_{i=0}^{\Ntsigma}\exp\Bigl(-\rho(\tau_{i},t-\tau_{i})\Bigr)} < +\infty$.

    Let a switching signal $\sigma$ satisfy
    \begin{enumerate}[leftmargin = *]
        \item $\rho^{\mathrm{S}}_{1}(r,s) = 0.2030s+0.0001s^{3/2}$, $\rho^{\mathrm{U}}_{2}(r,s) = 0.1s$,
        \item $\rho_{12}(r,s) = 0.1s+0.05s^{3/2}$, $\rho_{21}(r,s) = 0.2s+0.0025s^{3/2}$
    \end{enumerate}
    in addition to satisfying \eqref{issc_e:isslb}, \eqref{issc_e:nonissub}, and \eqref{issc_e:swub} for every interval $]r,r+s]\subset[0,+\infty[$ of time.

    We verify that
    \begin{align*}
        -\abs{\lambda_{1}}\rho^{\mathrm{S}}_{1}(r,s)&+\abs{\lambda_{2}}\rho^{\mathrm{U}}_{2}(r,s)+(\ln\mu_{12})\rho_{12}(r,s)\\
        &+(\ln\mu_{21})\rho_{21}(r,s) = -1.725\times 10^{-5}s^{3/2},
    \end{align*}
    which satisfies \eqref{issc_e:maincondn1} with $k_{1} = 1\times 10^{-5}$ and $c_{1} = 0$.

Let $\overline{\mathrm{T}}^{\mathrm{S}}_{1} = 0.01$, $\overline{\mathrm{T}}^{\mathrm{U}}_{2} = 2.58$, $\overline{\mathrm{N}}_{12} = \overline{\mathrm{N}}_{21} = 1$. An execution of the switching signal $\sigma$ described above is illustrated in Figure \ref{fig:swplot2}.
    \begin{figure}[htbp]
        \begin{center}
            \includegraphics[height = 6cm, width = 7cm]{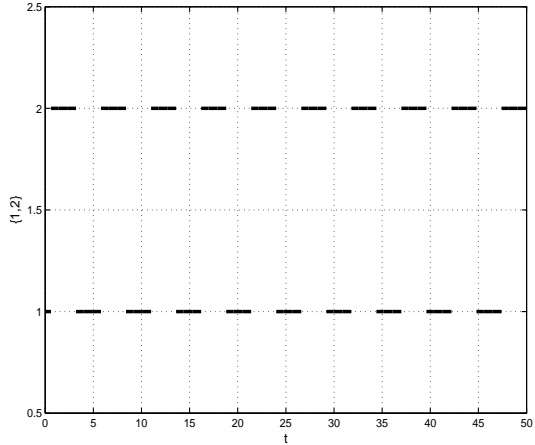}
            \caption{The switching signal for Figure \ref{fig:xtplot2}.} \label{fig:swplot2}
        \end{center}
        \end{figure}
    We study the process $(\norm{x(t)})_{t\geq 0}$ corresponding to fifty different initial conditions $x(0)$ selected uniformly at random from the interval $[-1000,1000]^{2}$ and the switching signal demonstrated in Figure \ref{fig:swplot2}. This is shown in Figure \ref{fig:xtplot2}.
    \begin{figure}[htbp]
        \begin{center}
           \includegraphics[height = 6cm, width = 7cm]{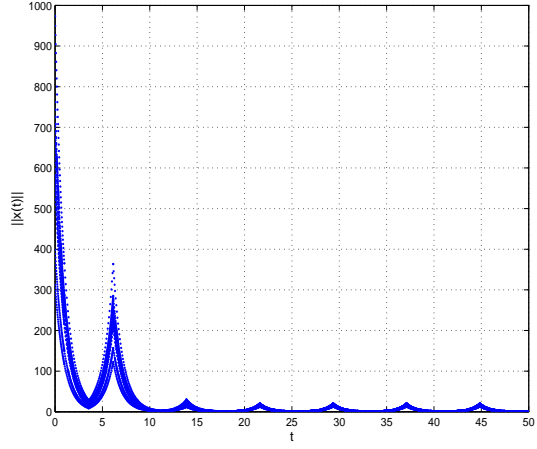}
            \caption{Plot of $\norm{x(t)}$ against $t$, with $x(0)$ selected uniformly at random from $[-1000,1000]^{2}$.}
            \label{fig:xtplot2}
        \end{center}
        \end{figure}

\section{Discussion}
\label{s:discussn}
    In this section we recast a subclass of average dwell time switching signals in our setting and establish analogs of the prior results: an ISS version of \cite[Theorem 2]{Liberzon_IOSS}, and \cite[Theorem 3.1]{chatterjee07}, with the aid of our main result. Our first result of this section is:
    \begin{proposition}
    \label{issc_cor:IOSS2012}
        Consider the family of systems \eqref{issc_e:family}. Suppose that Assumption \ref{issc_assumption:key} holds with $\abs{\lambda_{j}} = \lambda_{S}$ for all $j\in\P_{S}$ and $\abs{\lambda_{k}} = \lambda_{U}$ for all $k\in\P_{U}$, and Assumption \ref{issc_assumption:muijreln} holds with $\mu_{mn} = \mu$ for all $(m,n)\in E(\P)$. Let $\overline{\rho}$ and $\tau_{a}$ be constants satisfying $\displaystyle{\overline{\rho} \in\: \biggr]0,\frac{\lambda_{S}}{\lambda_{S}+\lambda_{U}}}\biggl[$ and
        \begin{align}
        \label{issc_e:adtmodified5}
            \tau_{a} \in\: \biggr]\frac{\ln\mu}{\lambda_{S}\cdot(1-\overline{\rho})-\lambda_{U}\cdot\overline{\rho}},+\infty\biggl[.
        \end{align}
        Let the class $\mathcal{F}\mathcal{K}_{\infty}$ functions $\rho^{\mathrm{S}}_{j}$, ${j\in\P_{S}}$, $\rho^{\mathrm{U}}_{k}$, ${k\in\P_{U}}$, $\rho_{mn}$, ${(m,n)\in E(\P)}$ described in Assumption \ref{issc_assumption:swsigbounds}, be such that for all $r_{1},r_{2}\geq 0$ and all $s > 0$
        \begin{align}
            \label{e:as1}\rho^{\mathrm{S}}_{j}(r_{1},s) &= \rho^{\mathrm{S}}_{j}(r_{2},s),\\
            \label{e:as2}\rho^{\mathrm{U}}_{k}(r_{1},s) &= \rho^{\mathrm{U}}_{k}(r_{2},s),
            \intertext{and}
            \label{e:as3}\rho_{mn}(r_{1},s) &= \rho_{mn}(r_{2},s).
        \end{align}
        Moreover, let for every interval $]r,r+s]\subset[0,+\infty[$ of time
        \begin{align}
            \label{issc_e:adtmodified3p}\sum_{j\in\P_{S}}\rho^{\mathrm{S}}_{j}(r,s) +  \sum_{k\in\P_{U}}\rho^{\mathrm{U}}_{k}(r,s) &\geq s,\\
            \label{issc_e:adtmodified3} \sum_{k\in\P_{U}}\rho^{\mathrm{U}}_{k}(r,s) &\leq \overline{\rho}\cdot s,\\
            \intertext{and}
            \label{issc_e:adtmodified4} \sum_{(m,n)\in E(\P)}\rho_{mn}(r,s) &\leq \frac{s}{\tau_{a}}.
        \end{align}
        Then the switched system \eqref{issc_e:swsys} is ISS for every switching signal $\sigma\in\mathcal{S}$ satisfying \eqref{issc_e:isslb}, \eqref{issc_e:nonissub}, and \eqref{issc_e:swub} for every interval $]r,r+s]\subset[0,+\infty[$ of time.
    \end{proposition}
    \begin{remark}
        Given a family of systems in which not all subsystems are input/output-to-state stable (IOSS), in \cite[Theorem 2]{Liberzon_IOSS} the authors identified a class of switching signals obeying the average dwell time property under which the resulting switched system is IOSS. Our Proposition \ref{issc_cor:IOSS2012} is an analog of an ISS version of \cite[Theorem 2]{Liberzon_IOSS} obtained as a corollary of our main result Theorem \ref{issc_t:mainres1}.
    \end{remark}
    \begin{remark}
        Since under average dwell time switching, the bounds on every interval $]r,r+s]\subset[0,+\infty[$ of time are independent of the initial point $r\in[0,+\infty[$ of the interval, the assumption that the class $\mathcal{F}\mathcal{K}_{\infty}$ functions $\rho^{\mathrm{S}}_{j}$, ${j\in\P_{S}}$, $\rho^{\mathrm{U}}_{k}$, ${k\in\P_{U}}$, $\rho_{mn}$, ${(m,n)\in E(\P)}$ satisfy \eqref{e:as1}-\eqref{e:as3} is natural.
    \end{remark}
    \begin{remark}
    \label{issc_r:adtremark}
         The bound on $\overline{\rho}$ ensures that $0<\overline{\rho}<1$. Consequently, the activation of unstable systems on every interval of time is restricted. A switching signal $\sigma$ that satisfies \eqref{issc_e:swub} on every interval $]r,r+s]\subset[0,+\infty[$ of time such that hypothesis \eqref{issc_e:adtmodified4} holds with $\rho_{mn}(r,s)$, ${(m,n)\in E(\P)}$ being independent of the first argument {implies that} the switching signal satisfies the average dwell time property \cite[p.\ 58]{Liberzon}. We have
        \begin{align*}
            \displaystyle\mathrm{N}_{\sigma}(s,t) = \sum_{(m,n)\in E(\P)}\mathrm{N}_{mn}(s,t) \leq \mnsum\overline{\mathrm{N}}_{mn} + \mnsum\rho_{mn}(s,t-s).
        \end{align*}
        Choose $\mathrm{N}_{0}$ such that $\mnsum\overline{\mathrm{N}}_{mn}\leq \mathrm{N}_{0}$. By hypothesis \eqref{issc_e:adtmodified4}, we have\\$\displaystyle{\mnsum\rho_{mn}(s,t-s) \leq \frac{t-s}{\tau_{a}}}$. Consequently, $\displaystyle{\mathrm{N}_{\sigma}(s,t) \leq \mathrm{N}_{0} + \frac{t-s}{\tau_{a}}}$ for positive constants $\mathrm{N}_{0}$ and $\tau_{a}$.
    \end{remark}
    A special case of \cite[Theorem 2]{Liberzon_IOSS} where all subsystems are ISS was treated in \cite[Theorem 3.1]{chatterjee07}. A subclass of average dwell time switching signals was proposed under which the resulting switched system is ISS. We recast an analog of \cite[Theorem 3.1]{chatterjee07} as a corollary of our main result:
    \begin{proposition}
    \label{issc_cor:ISS07}
        Consider the family of systems \eqref{issc_e:family}. Let $\P_{U} = \emptyset$. Suppose that Assumption \ref{issc_assumption:key} holds with $\abs{\lambda_{j}} = \lambda_{0}$ for all $j\in\P_{S}$ and Assumption \ref{issc_assumption:muijreln} holds with $\mu_{mn} = \mu$ for all $(m,n)\in E(\P)$. Let $\tau_{a}$ be a constant satisfying
        \begin{align}
        \label{issc_e:adtmodified2}
            \tau_{a} \in \biggr]\frac{\ln\mu}{\lambda_{0}},+\infty\biggl[.
        \end{align}
        Let the class $\mathcal{F}\mathcal{K}_{\infty}$ functions $\rho^{\mathrm{S}}_{j}$, $j\in\P_{S}$ and $\rho_{mn}$, $(m,n)\in E(\P)$ described in Assumption \ref{issc_assumption:swsigbounds} be such that for all $r_{1},r_{2} \geq 0$ and all $s > 0$
        \begin{align}
            \label{e:as4} \rho^{\mathrm{S}}_{j}(r_{1},s) &= \rho^{\mathrm{S}}_{j}(r_{2},s),
            \intertext{and}
            \label{e:as5} \rho_{mn}(r_{1},s) &= \rho_{mn}(r_{2},s).
        \end{align}
        Moreover, let for every interval $]r,r+s]\subset[0,+\infty[$ of time
        \begin{align}
        \label{issc_e:adtmodified0}
            \sum_{j\in\P_{S}}\rho^{\mathrm{S}}_{j}(r,s) &\geq s,\\
        \intertext{and}
        \label{issc_e:adtmodified1}
            \sum_{(m,n)\in E(\P)}\rho_{mn}(r,s) &\leq \frac{s}{\tau_{a}}.
        \end{align}
        Then the switched system \eqref{issc_e:swsys} is ISS for every switching signal $\sigma\in\mathcal{S}$ that for every interval $]r,r+s]\subset[0,+\infty[$ of time, satisfies \eqref{issc_e:isslb} and \eqref{issc_e:swub}.
    \end{proposition}
    \begin{remark}
        Since $\P_{U} = \emptyset$, condition \eqref{issc_e:nonissub} is automatically satisfied. A switching signal that satisfies \eqref{issc_e:swub} such that \eqref{issc_e:adtmodified1} holds implies that the switching signal satisfies the average dwell time property as explained in Remark \ref{issc_r:adtremark}.
    \end{remark}


\section{Concluding remarks}
\label{s:concln}
    In this article we presented a class of switching signals under which a continuous-time switched system is uniformly ISS. We utilized multiple ISS-Lyapunov-like functions for our analysis and our characterization of stabilizing switching signals allowed the number of switches on any interval of time to grow faster than an affine function of the length of the interval unlike in the case of average dwell time switching. We also discussed two representative prior results: an ISS version of \cite[Theorem 2]{Liberzon_IOSS}, and \cite[Theorem 2]{chatterjee07} in our setting. Our results extend readily to the discrete-time setting.

\section{Proofs}
\label{s:proofs}
    \begin{proof}[Proof of Theorem \ref{issc_t:mainres1}]
        Fix $t > 0$. Then $0 =: \tau_{0}<\tau_{1}<\cdots<\tau_{\Ntsigma}$ are the switching instants before (and including) $t$. In view of \eqref{issc_e:dLyapineq},
        \begin{align}
            V_{\sigma(t)}(x(t)) &\leq \exp\bigl(-\lambda_{\sigma(\tau_{\Ntsigma})}(t-\tau_{\Ntsigma})\bigr)V_{\sigma(t)}(x(\tau_{\Ntsigma})) \nonumber\\
            \label{issc_e:proof1} &\qquad +\supnormterm\int_{\tau_{\Ntsigma}}^{t}\exp\bigl(-\lambda_{\sigma(\tau_{\Ntsigma})}(t-s)\bigr)ds.
        \end{align}
        Applying \eqref{issc_e:muijineq} and iterating the above, we obtain the estimate
        \begin{align}
        \label{issct_e:est}
            V_{\sigma(t)}(x(t)) &\leq \psi_{1}(t)V_{\sigma(0)}(x_{0})+\gamma(\norm{v}_{[0,t]})\psi_{2}(t),
        \end{align}
        where
        \begin{align}
        \label{issc_e:psi1defn}
            \psi_{1}(t) := \exp\left(-\ilsum\lambda_{\sigma(\tau_{i})}S_{i+1}+\imsum\ln\muiterm\right),
        \end{align}
        and
        \begin{align}
            \psi_{2}(t) &:= \ilsum\left(\exp\left(-\klsum\lambda_{\sigma(\tau_{k})}S_{k+1}+\kmsum\ln\mukterm\right)\right. \nonumber\\
            \label{issc_e:psi2defn} &\qquad\left.\times\frac{1}{\lambda_{\sigma(\tau_{i})}}\vphantom{\klsum}\bigl(1-\exp\bigl(-\lambda_{\sigma(\tau_{i})}S_{i+1}\bigr)\bigr)\right).
        \end{align}
        In view of \eqref{issc_e:Lyapineq} we rewrite the estimate \eqref{issct_e:est} as
        \begin{align*}
            \underline{\alpha}(\norm{x(t)})\leq \psi_{1}(t)\overline{\alpha}(\norm{x_{0}})+\supnormterm\psi_{2}(t).
        \end{align*}
        In view of Definition \ref{issc_d:iss} for ISS of \eqref{issc_e:swsys}, we need to first show the following:
        \begin{enumerate}[label = \roman*), leftmargin = *]
            \item $\overline{\alpha}(*)\psi_{1}(\cdot)$ can be bounded above by a class $\KL$ function, and
            \item $\psi_{2}(\cdot)$ is bounded by a constant, say $\overline{\psi}_{2}$.
        \end{enumerate}

         The function $\psi_{1}(t)$ is
        \begin{align}
            &\exp\left(-\sum_{j\in\P_{S}}\abs{\lambda_{j}}\ztmeasure+\sum_{k\in\P_{U}}\abs{\lambda_{k}}\ztkmeasure\right.\nonumber\\
            &\qquad +\left.\vphantom{\stmeasure}\mnsum(\ln\mu_{mn})\mncount_{0}^{\tau_{\Ntsigma}-1}\right) \nonumber\\
            \label{issc_e:proof3} & =\exp\Biggl(-\sum_{j\in\P_{S}}\abs{\lambda_{j}}\mathrm{T}^{\mathrm{S}}_{j}(0,t)+\sum_{k\in\P_{U}}\abs{\lambda_{k}}\mathrm{T}^{\mathrm{U}}_{k}(0,t)+\mnsum(\ln\mu_{mn})\mathrm{N}_{mn}(0,\tau_{\Ntsigma-1})\Biggr),
        \end{align}
        and $\psi_{2}(t)$ is
        \begin{align}
            &\sum_{j\in\P_{S}}\frac{1}{\abs{\lambda_{j}}}\ijsum\Biggl(\exp\biggl(-\sum_{p\in\P_{S}}\abs{\lambda_{p}}\mathrm{T}^{\mathrm{S}}_{p}(\tau_{i+1},t)+\sum_{q\in\P_{U}}\abs{\lambda_{q}}\mathrm{T}^{\mathrm{U}}_{q}(\tau_{i+1},t) \nonumber\\
            &+\mnsum(\ln\mu_{mn})\mathrm{N}_{mn}(\tau_{i+1},\tau_{\Ntsigma-1})\biggr)\biggl(1-\exp\bigl(-\abs{\lambda_{j}}S_{i+1}\bigr)\biggr)\Biggr) \nonumber\\
            &+\sum_{k\in\P_{U}}\frac{1}{\abs{\lambda_{k}}}\iksum\Biggl(\exp\biggl(-\sum_{p\in\P_{S}}\abs{\lambda_{p}}\mathrm{T}^{\mathrm{S}}_{p}(\tau_{i+1},t)+\sum_{q\in\P_{U}}\abs{\lambda_{q}}\mathrm{T}^{\mathrm{U}}_{q}(\tau_{i+1},t) \nonumber\\
            &+\mnsum(\ln\mu_{mn})\mathrm{N}_{mn}(\tau_{i+1},\tau_{\Ntsigma-1})\biggr)\biggl(1-\exp\bigl(\abs{\lambda_{k}}S_{i+1}\bigr)\biggr)\Biggr)\nonumber\\
            \leq& \sum_{j\in\P_{S}}\frac{1}{\abs{\lambda_{j}}}\ijsum\Biggl(\exp\biggl(-\sum_{p\in\P_{S}}\abs{\lambda_{p}}\mathrm{T}^{\mathrm{S}}_{p}(\tau_{i+1},t)+\sum_{q\in\P_{U}}\abs{\lambda_{q}}\mathrm{T}^{\mathrm{U}}_{q}(\tau_{i+1},t) \nonumber\\
            &+\mnsum(\ln\mu_{mn})\mathrm{N}_{mn}(\tau_{i+1},t)\biggr)\nonumber\\+ &\sum_{k\in\P_{U}}\frac{1}{\abs{\lambda_{k}}}\iksum\Biggl(\exp\biggl(-\sum_{p\in\P_{U}}\abs{\lambda_{p}}\mathrm{T}^{\mathrm{S}}_{p}(\tau_{i},t) +\sum_{q\in\P_{U}}\abs{\lambda_{q}}\mathrm{T}^{\mathrm{U}}_{q}(\tau_{i},t)\nonumber\\
            \label{issc_e:proof4} &+\mnsum(\ln\mu_{mn})\mathrm{N}_{mn}(\tau_{i},t)\biggr) .
            \end{align}
        By hypotheses \eqref{issc_e:isslb}, \eqref{issc_e:nonissub}, and \eqref{issc_e:swub}, we have the right-hand side of \eqref{issc_e:proof4} bounded above by
        \begin{align*}
            &\sum_{j\in\P_{S}}\frac{1}{\abs{\lambda_{j}}}\ijsum\exp\biggl(\sum_{p\in\P_{S}}\abs{\lambda_{p}}\bigl(\overline{\mathrm{T}}^{\mathrm{S}}_{p}-\rho^{\mathrm{S}}_{p}(\tau_{i+1},t-\tau_{i+1})\bigr)&\nonumber\\
            +&\sum_{q\in\P_{U}}\abs{\lambda_{q}}\bigl(\overline{\mathrm{T}}^{\mathrm{U}}_{q}+\rho^{\mathrm{U}}_{q}(\tau_{i+1},t-\tau_{i+1})
            +\mnsum(\ln\mu_{mn})\bigl(\overline{\mathrm{N}}_{mn}+\rho_{mn}(\tau_{i+1},t-\tau_{i+1})\bigr)\biggr) \nonumber\\
            +&\sum_{k\in\P_{U}}\frac{1}{\abs{\lambda_{k}}}\iksum\exp\biggl(\sum_{p\in\P_{S}}\abs{\lambda_{p}}\bigl(\overline{\mathrm{T}}^{\mathrm{S}}_{p}-\rho^{\mathrm{S}}_{p}(\tau_{i},t-\tau_{i})\bigr)\\
            +&\sum_{q\in\P_{U}}\abs{\lambda_{q}}\bigl(\overline{\mathrm{T}}^{\mathrm{U}}_{q}+\rho^{\mathrm{U}}_{q}(\tau_{i},t-\tau_{i})+\mnsum(\ln\mu_{mn})\bigl(\overline{\mathrm{N}}_{mn}+\rho_{mn}(\tau_{i},t-\tau_{i})\bigr)\biggr).
        \end{align*}
        By \eqref{issc_e:maincondn1}, the above expression is bounded above by
        \begin{align}
            & \left(\sum_{j\in\P_{S}}\frac{1}{\abs{\lambda_{j}}}\ijsum\exp\Bigl(c+c_{1}-\rho(\tau_{i+1},t-\tau_{i+1})\Bigr)\right.\nonumber\\
            &\quad\left.+\sum_{k\in\P_{U}}\frac{1}{\abs{\lambda_{k}}}\iksum\exp\Bigl(c+c_{1}-\rho(\tau_{i},t-\tau_{i})\Bigr)\right) \nonumber\\
            \label{issc_e:proof5} \leq& \left(\sum_{j\in\P_{S}}\frac{1}{\abs{\lambda_{j}}}\sum_{i=0}^{\Ntsigma}\exp\bigl(c+c_{1}-\rho(\tau_{i+1},t-\tau_{i+1})\bigr)\right.\nonumber\\&\quad\left.+\sum_{k\in\P_{U}}\frac{1}{\abs{\lambda_{k}}}\sum_{i=0}^{\Ntsigma}\exp\bigl(c+c_{1}-\rho(\tau_{i},t-\tau_{i})\bigr)\right),
        \end{align}
        for some $c > 0$ satisfying
        \[
            \sum_{j\in\P_{S}}\overline{\mathrm{T}}^{\mathrm{S}}_{j}+\sum_{k\in\P_{U}}\overline{\mathrm{T}}^{\mathrm{U}}_{k}+\sum_{(m,n)\in E(\P)}\overline{\mathrm{N}}_{mn} \leq c.
        \]
        In view of \eqref{issc_e:maincondn2} and the fact that $\P$ is finite, both the terms
        \[
            \displaystyle\sum_{j\in\P_{S}}\frac{1}{\abs{\lambda_{j}}}\sum_{i=0}^{\Ntsigma}\exp\bigl(c+c_{1}-\rho(\tau_{i+1},t-\tau_{i+1})\bigr)\quad\text{and}\quad\displaystyle\sum_{k\in\P_{U}}\frac{1}{\abs{\lambda_{k}}}\sum_{i=0}^{\Ntsigma}\exp\bigl(c+c_{1}-\rho(\tau_{i},t-\tau_{i})\bigr)
        \]
        are bounded. Consequently, ii) holds. It remains to verify i). Towards this end, we already see that $\overline{\alpha}\in\Kinfty$ from Assumption \ref{issc_assumption:key}. Therefore, it remains to show that $\psi_{1}(\cdot)$ is bounded above by a function in class $\mathcal{L}$ to complete the proof of i).\footnote{$\mathcal{L} := \bigl\{\gamma:[0,+\infty[\lra[0,+\infty[\:\:\big|\:\:\gamma\:\:\text{is continuous and}\:\:\gamma(s)\searrow 0\:\:\text{as}\:\:s\nearrow +\infty\bigr\}$} By hypotheses \eqref{issc_e:isslb}, \eqref{issc_e:nonissub}, and \eqref{issc_e:swub}, we have $\psi_{1}(t)$ is bounded above by
        \begin{align*}
            \exp& \Biggl(\sum_{j\in\P_{S}}\abs{\lambda_{j}}(\overline{\mathrm{T}}^{\mathrm{S}}_{j}-\rho^{\mathrm{S}}_{j}(0,t))+\sum_{k\in\P_{U}}\abs{\lambda_{k}}(\overline{\mathrm{T}}^{\mathrm{U}}_{k}+\rho^{\mathrm{U}}_{k}(0,t))\\ &+\mnsum(\ln\mu_{mn})(\overline{\mathrm{N}}_{mn}+\rho_{mn}(0,t))\Biggr).
        \end{align*}
        By \eqref{issc_e:maincondn1} the above quantity is at most $\exp\bigl(c+c_{1}-\rho(0,t)\bigr)$,
        which decreases as $t$ increases, and tends to $0$ as $t\tendsto+\infty$. To summarize,
        \[
        \alpha(\norm{x(t)}) \leq \beta(\norm{x_{0}},t) + \chi(\norm{v}_{[0,t]})\:\:\text{for all}\:\: t\geq 0
        \]
        holds with $\alpha(r) := r$, $\displaystyle{\beta(r,s) = \overline{\alpha}(r)\exp\bigl(c+c_{1}-\rho(0,s)\bigr)}$ and $\chi(r):= \gamma(r)\overline{\psi}_{2}$, where
        \begin{align*}
            \overline{\psi}_{2} &= \biggl(\displaystyle\sum_{j\in\P_{S}}\frac{1}{\abs{\lambda_{j}}}\sup_{t}\sum_{i=0}^{\Ntsigma}\exp\bigl(c+c_{1}-\rho(\tau_{i+1},t-\tau_{i+1})\bigr)\\
            &\quad+\displaystyle\sum_{k\in\P_{U}}\frac{1}{\abs{\lambda_{k}}}\sup_{t}\sum_{i=0}^{\Ntsigma}\exp\bigl(c+c_{1}-\rho(\tau_{i},t-\tau_{i})\bigr)\biggr).
            \end{align*}
        This completes our proof for ISS. For uniformity over $\sigma$, we note that the functions $\beta$ and $\chi$ do not depend on the specific switching signal $\sigma$ satisfying \eqref{issc_e:isslb}, \eqref{issc_e:nonissub}, and \eqref{issc_e:swub} under our assumptions.
    \end{proof}

    \begin{proof}[Proof of Lemma \ref{issc_lem:rhoaffine}]
    We express $\rho_{\mathrm{N}}^{-1}(\cdot,t-\cdot)(1)$ by $\rho_{\mathrm{N}}^{-1}(1)$ for notational simplicity. We have
    \begin{align*}
        \sum_{i=0}^{\Ntsigma}&\exp\bigl(-\rho(\tau_{i},t-\tau_{i})\bigr) \leq \sum_{i=0}^{\mathrm{N}_{0}+\lfloor\rho_{\mathrm{N}}(0,t)\rfloor}\exp\bigl(-\rho(\tau_{i},t-\tau_{i})\bigr)\\
        &= \exp(-k_{2})\sum_{i=0}^{\mathrm{N}_{0}+\lfloor\rho_{\mathrm{N}}(0,t)\rfloor}\exp\bigl(-k_{1}\cdot(t-\tau_{i})\bigr)\\
        &= \exp(-k_{2})\biggl(\exp\bigl(-k_{1}\cdot(t)\bigr)+\exp\bigl(-k_{1}\cdot(t-\tau_{1})\bigr)+\cdots\\
        &\quad+\exp\bigl(-k_{1}\cdot(t-\tau_{\lfloor\rho_{\mathrm{N}}(0,t)\rfloor})+\exp\bigl(-k_{1}\cdot(t-\tau_{\lfloor\rho_{\mathrm{N}}(0,t)\rfloor+1})\bigr)+\cdots\\
        &\quad+\exp\bigl(-k_{1}\cdot(t-\tau_{\lfloor\rho_{\mathrm{N}}(0,t)\rfloor+\mathrm{N}_{0}})\bigr)\biggr)\\
        &= \exp(-k_{2})\biggl(1+\mathrm{N}_{0}+\exp(-nk_{1}\rho_{\mathrm{N}}^{-1}(1))\\&\quad+\exp(-(n-1)k_{1}\rho_{\mathrm{N}}^{-1}(1))+\cdots\\&\quad+\exp(-2k_{1}\rho_{\mathrm{N}}^{-1}(1))+\exp(-k_{1}\rho_{\mathrm{N}}^{-1}(1))\biggr)\\
        &\leq \exp(-k_{2})\Biggl(1+\mathrm{N}_{0}+\frac{1}{\exp(-k_{1}\rho_{\mathrm{N}}^{-1}(1))-1}\Biggr).\qedhere
    \end{align*}
    \end{proof}

    \begin{proof}[Proof of Lemma \ref{issc_lem:rhothreehalves}]
    We have
    \begin{align*}
        \sum_{i=0}^{N_{\sigma}(0,t)}&\exp\bigl(-\rho(\tau_{i},t-\tau_{i})\bigr)  \leq \sum_{i=0}^{\mathrm{N}_{0}+\lfloor\rho_{N}(0,t)\rfloor}\exp\bigl(-\rho(\tau_{i},t-\tau_{i})\bigr)\\
        &\leq \exp(-k_{2})\Biggl((\mathrm{N}_{0}+1)+\exp\bigl(-k_{1}(\rho_{N}^{-1}(1))^{3/2}n^{3/2}\bigr)
        +\cdots\\
        &\quad+\exp\bigl(-k_{1}(\rho_{N}^{-1}(1))^{3/2}2^{3/2}\bigr)+\exp\bigl(-k_{1}(\rho_{N}^{-1}(1))^{3/2}\bigr)\Biggr).
    \end{align*}
    We apply the integral test \cite[\S3.3]{KaczorSeries}; we define a new variable $y^{2} := x^{3}$, and compute
    \begin{align*}
        \int_{0}^{+\infty}\exp\bigl(-k_{1}(\rho_{N}^{-1}(1))^{3/2}\bigr)dx &= \frac{2}{3}\int_{0}^{+\infty}y^{-1/3}\exp\bigl(-k_{1}(\rho_{N}^{-1}(1))^{3/2}y\bigr)dy\\
        &= \frac{2}{3k_{1}(\rho_{N}^{-1}(1))^{3/2}}\Gamma\Biggl(\frac{2}{3}\Biggr),
    \end{align*}
    which is finite, showing thereby that $\displaystyle{\sum_{i=0}^{\Ntsigma}\exp\bigl(-\rho(\tau_{i},t-\tau_{i})\bigr)}$ is bounded.
    \end{proof}

    \begin{proof}[Proof of Proposition \ref{issc_cor:IOSS2012}]
    Consider the left-hand side of \eqref{issc_e:maincondn1}. For every interval $]s,t]\subset[0,+\infty[$ of time, we have
    \begin{align*}
        -\sum_{j\in\P_{S}}\abs{\lambda_{j}}\rho^{\mathrm{S}}_{j}(s,t-s) + \sum_{k\in\P_{U}}\abs{\lambda_{k}}\rho^{\mathrm{U}}_{k}(s,t-s)
        +\sum_{(m,n)\in E(\P)}(\ln\mu_{mn})\rho_{mn}(s,t-s).
    \end{align*}
    By hypotheses $\abs{\lambda_{j}} = \lambda_{S}$ for all $j\in\P_{S}$, $\abs{\lambda_{k}} = \lambda_{U}$ for all $k\in\P_{U}$, and $\mu_{mn}=\mu$ for all $(m,n)\in E(\P)$. Consequently, the above quantity is equal to
    \begin{align}
    \label{issc_e:corpf1}
        -\lambda_{S}\sum_{j\in\P_{S}}\rho^{\mathrm{S}}_{j}(s,t-s) + \lambda_{U}\sum_{k\in\P_{U}}\rho^{\mathrm{U}}_{k}(s,t-s) + (\ln\mu)\sum_{(m,n)\in E(\P)}\rho_{mn}(s,t-s).
    \end{align}
    By hypothesis \eqref{issc_e:adtmodified3p}, \eqref{issc_e:adtmodified3}, and \eqref{issc_e:adtmodified4}, the above quantity is at most equal to
    \begin{align}
    \label{issc_e:corpf2}
        -\lambda_{S}\cdot(1-\overline{\rho})\cdot(t-s) + \lambda_{U}\cdot\overline{\rho}\cdot(t-s) + (\ln\mu)\cdot\frac{t-s}{\tau_{a}}.
    \end{align}
    By \eqref{issc_e:adtmodified5},
    \begin{align}
    \label{issc_e:corpf3}
        \frac{1}{\tau_{a}} \leq \frac{\lambda_{S}\cdot(1-\overline{\rho})-\lambda_{U}\cdot\overline{\rho}}{\ln\mu} - \varepsilon\:\:\text{for some}\:\:\varepsilon > 0.
    \end{align}
    From \eqref{issc_e:corpf3}, we have that \eqref{issc_e:corpf2} is bounded above by
    \begin{align*}
        &-\lambda_{S}\cdot(1-\overline{\rho})\cdot(t-s) + \lambda_{U}\cdot\overline{\rho}\cdot(t-s) + (\ln\mu)\cdot\frac{(\lambda_{S}\cdot(1-\overline{\rho})-\lambda_{U}\cdot\overline{\rho})}{(\ln\mu)}\cdot(t-s)\\
        &\quad- (\ln\mu)\cdot\varepsilon\cdot(t-s)\\
        &= -\lambda_{S}\cdot(1-\overline{\rho})\cdot(t-s) + \lambda_{U}\cdot\overline{\rho}\cdot(t-s) + \lambda_{S}\cdot(1-\overline{\rho})\cdot(t-s)\\&\quad- \lambda_{U}\cdot\overline{\rho}\cdot(t-s) - (\ln\mu)\cdot\varepsilon\cdot(t-s)\\
        &= -(\ln\mu)\cdot\varepsilon\cdot(t-s,)
    \end{align*}
    which is equivalent to $c_{1}-\rho(s,t-s)$ with $c_{1} = 0$ and $\rho$ linear in the second argument.

    \emph{Claim}: The series $\displaystyle{\sum_{i=0}^{\Ntsigma}}\exp\bigl(-\varepsilon\cdot(\ln\mu)\cdot(t-\tau_{i})\bigr)$ for some $\varepsilon > 0$, is bounded with respect to $t$ under average dwell time switching.\\
        Let $\varepsilon' = \varepsilon\cdot(\ln\mu)$. We consider the worst case switching identified in Remark \ref{issct_r:adtcompa}.
    \begin{align*}
        \sum_{i=0}^{\mathrm{N}_{0}+\lfloor\frac{t}{\tau_{a}}\rfloor}&\exp\bigl(-\varepsilon'\cdot(t-\tau_{i+1})\bigr) \leq \exp\bigl(-\varepsilon'\cdot(t-\tau_{1})\bigr) + \exp\bigl(-\varepsilon'\cdot(t-\tau_{2})\bigr) + \cdots\\
        &\quad+ \exp\bigl(-\varepsilon'\cdot(t-\tau_{\lfloor\frac{t}{\tau_{a}}\rfloor})\bigr) + \mathrm{N}_{0}\\
        &= \exp\biggl(t-\bigl(t-\lfloor\frac{t}{\tau_{a}}\rfloor\tau_{a}\bigr)\biggr) + \exp\biggl(t-(t-\bigl(\lfloor\frac{t}{\tau_{a}}\rfloor-1\bigr)\tau_{a})\biggr)\\
        & \quad+ \cdots + \exp\bigl(t-(t-2\tau_{a})\bigr) + \exp\bigl(t-(t-\tau_{a})\bigr) + \mathrm{N}_{0}\\
        &= \exp(-\varepsilon' t) + \exp\biggl(-\varepsilon'\bigl(\lfloor\frac{t}{\tau_{a}}\rfloor-1\bigr)\tau_{a}\biggr) + \cdots + \exp\bigl(-2\varepsilon'\tau_{a}\bigr)\\
        & \quad+ \exp\bigl(-\varepsilon'\tau_{a}\bigr) + \mathrm{N}_{0}\\
        &\leq 1 + \mathrm{N}_{0} + \exp(-\varepsilon'\tau_{a})\frac{1-(\exp(-\varepsilon'\tau_{a}))^{\lfloor\frac{t}{\tau_{a}}+1\rfloor}}{1-\exp(-\varepsilon'\tau_{a})}\\
        &\leq 1 + \mathrm{N}_{0} + \frac{1}{\exp(\varepsilon'\tau_{a})-1}.
    \end{align*}
    This proves our claim and the assertion of Theorem \ref{issc_t:mainres1} follows at once.

\end{proof}

    \begin{proof}[Proof of Proposition \ref{issc_cor:ISS07} (Sketch)]
        Observe that under the hypothesis $\P_{U} = \emptyset$, for every interval $]s,t]\subset[0,+\infty[$ of time, the left-hand side of \eqref{issc_e:maincondn1} becomes
        \begin{align*}
        -\sum_{j\in\P_{S}}\abs{\lambda_{j}}\rho^{\mathrm{S}}_{j}(s,t-s)
        +\sum_{(m,n)\in E(\P)}(\ln\mu_{mn})\rho_{mn}(s,t-s).
    \end{align*}
    The rest of the proof for Proposition \ref{issc_cor:ISS07} follows under the same set of arguments as in the proof of Proposition \ref{issc_cor:IOSS2012}.
    \end{proof}


\begin{thebibliography}{10}

\bibitem{Angeli99}
{\sc D.~Angeli and E.~D. Sontag}, {\em Forward completeness, unboundedness
  observability, and their {L}yapunov characterizations}, Systems Control
  Lett., 38 (1999), pp.~209--217.

\bibitem{DePersis2003}
{\sc C.~De~Persis, R.~De~Santis, and A.~S. Morse}, {\em Switched nonlinear
  systems with state-dependent dwell-time}, Systems Control Lett., 50 (2003),
  pp.~291--302.

\bibitem{Filippov}
{\sc A.~F. Filippov}, {\em Differential equations with discontinuous righthand
  sides}, vol.~18 of Mathematics and its Applications (Soviet Series), Kluwer
  Academic Publishers Group, Dordrecht, 1988.
\newblock Translated from the Russian.

\bibitem{HespanhaMorse}
{\sc J.~P. Hespanha and A.~S. Morse}, {\em Stability of switched systems with
  average dwell-time}, in Proc. of the 38th Conf. on Decision and Contr., Dec
  1999, pp.~2655--2660.

\bibitem{KaczorSeries}
{\sc W.~J. Kaczor and M.~T. Nowak}, {\em Problems in Mathematical Analysis I},
  vol.~4, American Mathematical Society, 2000.
\newblock Real Numbers, Sequences and Series.

\bibitem{Krichman01}
{\sc M.~Krichman, E.~D. Sontag, and Y.~Wang}, {\em Input-output-to-state
  stability}, SIAM J. Control Optim., 39 (2001), pp.~1874--1928 (electronic).

\bibitem{TACsub}
{\sc A.~Kundu and D.~Chatterjee}, {\em Stabilizing switching signals for
  switched systems}.
\newblock To appear in IEEE Transactions on Automatic Control, doi:
  10.1109/TAC.2014.2335291.

\bibitem{Liberzon}
{\sc D.~Liberzon}, {\em Switching in systems and control}, Systems \& Control:
  Foundations \& Applications, Birkh\"auser Boston Inc., Boston, MA, 2003.

\bibitem{Morse1996}
{\sc A.~S. Morse}, {\em Supervisory control of families of linear set-point
  controllers. {I}. {E}xact matching}, IEEE Trans. Automat. Control, 41 (1996),
  pp.~1413--1431.

\bibitem{Liberzon_IOSS}
{\sc M.~A. M{\"u}ller and D.~Liberzon}, {\em Input/output-to-state stability
  and state-norm estimators for switched nonlinear systems}, Automatica J.
  IFAC, 48 (2012), pp.~2029--2039.

\bibitem{Sontag95}
{\sc E.~D. Sontag and Y.~Wang}, {\em On characterizations of the input-to-state
  stability property}, Systems Control Lett., 24 (1995), pp.~351--359.

\bibitem{chatterjee07}
{\sc L.~Vu, D.~Chatterjee, and D.~Liberzon}, {\em Input-to-state stability of
  switched systems and switching adaptive control}, Automatica, 43 (2007),
  pp.~639--646.

\bibitem{XieWenLi2001}
{\sc W.~Xie, C.~Wen, and Z.~Li}, {\em Input-to-state stabilization of switched
  nonlinear systems}, IEEE Trans. Automat. Control, 46 (2001), pp.~1111--1116.

\bibitem{Yang14}
{\sc G.~Yang and D.~Liberzon}, {\em Input-to-state stability for switched
  systems with unstable subsystems: A hybrid lyapunov construction}, in Proc.
  of the 53rd Conf. on Decision and Contr., Dec 2014, pp.~6240--6245.

\end{thebibliography}

\end{document}